\documentclass[11pt,twocolumn]{article}

\usepackage{anysize}
\usepackage[utf8x]{inputenc}
\usepackage{tikz}
\usepackage{amssymb,amsmath}
\usepackage{mathtools}
\usepackage{multirow}
\usepackage{listings}
\usepackage{alltt}
\usepackage{fancybox}
\tikzstyle{block} = [rectangle, draw, text width=5em,text centered, minimum height=2em]
\usepackage{pstricks}
\usepackage{pst-coil}
\usepackage{rotating}
\psset{unit=1mm}
\usepackage{pgf,pgfarrows,pgfnodes,pgfautomata,pgfheaps}
\usepackage{etex}
\usepackage{colortbl}
\usepackage{pst-node}
\usepackage{pst-grad}
\usepackage{pst-text}
\usepackage{amssymb}
\usepackage[breaklinks=true]{hyperref}
\usepackage[normalem]{ulem}
\usepackage{fancyhdr}
\hypersetup{colorlinks, linkcolor=black, urlcolor=blue, pdftex, pagebackref=true, citecolor=black
}
\usepackage{graphicx}
\usepackage{epstopdf}
\usepackage{epsfig}
\usepackage{epsfig}
\usepackage{graphicx}
\usepackage{epstopdf}
\usepackage{epsfig}
\usepackage{stfloats}
\usepackage{amssymb, amsmath, amsthm}
\usepackage{tabularx}
\usepackage{graphics}
\usepackage{verbatim}
\usepackage{algorithm}
\usepackage{algpseudocode}
\usepackage{listings}
\theoremstyle{plain}
\newtheorem{thm}{Theorem}[section]
\newtheorem{lem}[thm]{Lemma}

\theoremstyle{definition}

\theoremstyle{remark}

\begin{document}

\setlength{\pdfpageheight}{\paperheight}
\setlength{\pdfpagewidth}{\paperwidth}






\title{Parallelization of Loops with Variable Distance Data Dependences}
\author{Archana Kale \thanks{archanak@cse.iitb.ac.in}\and Amitkumar Patil \thanks{amitkumar@cse..iitb.ac.in} \and Supratim Biswas \thanks{sb@cse.iitb.ac.in}
        \\ Department of Computer Science and Engineering  
        \\ Indian Institute of Technology, Bombay}
\date{}



\maketitle

\begin{abstract}
The extent of parallelization of a loop is largely determined by the dependences between its statements. While dependence free loops are fully parallelizable, those with loop carried dependences are not. Dependence distance is a measure of absolute difference between a pair of dependent iterations. Loops with constant distance data dependence($CD^{3}$), because of uniform distance between the dependent iterations, lead to easy partitioning of the iteration space and hence they have been successfully dealt with.   
\\
Parallelization of loops with variable distance data dependences($VD^{3}$) is a considerably difficult problem. It is our belief that partitioning the iteration space in such loops cannot be done without examining solutions of the corresponding Linear Diophantine Equations(LDEs). Focus of this work is to study $VD^{3}$ and examine the relation between dependent iterations. Our analysis based on parametric solution leads to a mathematical formulation capturing dependence between iterations. Our approach shows the existence of reasonable exploitable parallelism in $VD^{3}$ loops with multiple LDEs.

\end{abstract}

\textbf{Category D.3.4 Code Generation, Compilers, Optimization}

\textbf{General Terms Dependence Distance, Iteration Space, Partition, Parametric Solutions, Linear Diophantine Equation, Parallelism, Algorithms, Schedule}

\section{Introduction}
Loops with loop carried $CD^{3}$s having cycle free Data Dependence Graph(DDG) have been successfully dealt with in the literature\cite{allen}, however few researchers have considered loops having cycles in DDG\cite{poly}. Fewer efforts have considered loops with variable distance data dependence($VD^{3}$) and DDGs with or without cycles\cite{cycle}. Unlike the case of constant distance, the solutions of LDEs for $VD^{3}$ do not seem to have a regular structure among dependent iterations. This is probably the reason for not enough research reported in this area.
\\
Our approach builds on a two variable LDE for which parametric solutions are well known \cite{exact}. We have analyzed the parametric solutions and have developed a mathematical formulation that captures the structure among dependent iterations. The structure leads to partitioning of the iteration space where each component of the partition represents iterations which have to be executed sequentially and distinct components are parallelizable. Further, bounds on the number and size of components have been obtained. Examining the structure of a component, we show that all iterations of the component can be generated from a single iteration, called its seed. The representation of the partition is consequently reduced to a set of seeds. This set is used as a basis to generate the components dynamically in a demand driven manner which can lead to schedules.
\\
We have extended this approach to partition the iteration space for multiple LDEs by combining the components of partitions of the individual LDEs. The correctness of the non-trivial composition of partitions to form the resultant partition is proved. We have presented algorithms for all important formulations of our work. The applicability of our approach rests on its ability to extract exploitable parallelism from $VD^{3}$ loops with multiple LDEs. Experimental evaluation of effectiveness of our approach requires existence of $VD^{3}$ loops with large number of LDEs. Since we could not find such loops in the benchmark programs, we chose to create multiple LDEs with random coefficients for our experimentation. We present an algorithm to generate iteration schedule. Results show that loops with $VD^{3}$ offer reasonable parallelism which reduces as the number of LDEs increase. To increase the parallelism in case of large number of LDEs we have formed an alternate partition using a heuristic. This heuristic shows non trivial improvement in parallelism.

\section {Motivation}
The approach is motivated by the following example. Consider a loop having $CD^{3}$ over iteration space R=[-8,7].
\begin {verbatim}
for(i = -8 ; i < 8 ; i++)
{ S1: g[i] = .........;
  S2: ......... = g[i + 3];}
\end{verbatim}
The $CD^{3}$ is "+3" and leads to a 3 element partitioning of iteration space for parallelization as follows: \{ \{-8,-5,-2,1,4,7\}, \{-7,-4,-1,2,5\}, \{-6,-3,0,3,6\} \}. \\
Consider the following $VD^{3}$ loop.
\begin {verbatim}
for(i = -8 ; i < 8 ; i++)
{ S1: g[2i + 1] = .........;
  S2: ......... = g[3i + 6];}
\end{verbatim} 
The dependent iteration pairs of the loop are modelled by LDE: $2x-3y = 5$. The solutions of the LDE i.e. dependent iteration pairs over R=[-8,7] are: \{ (-8,-7) , (-5,-5) , (-2,-3) , (1,-1) , (4,1) , (7,3) \}. This leads to Partition, $P=\cup_{\forall i} p_{i}$, of R as shown in Figure \ref{Partition} where iterations are nodes and edges represent dependence between them.
\begin{figure}[!htb]
\begin{center}
\includegraphics[scale =.4]{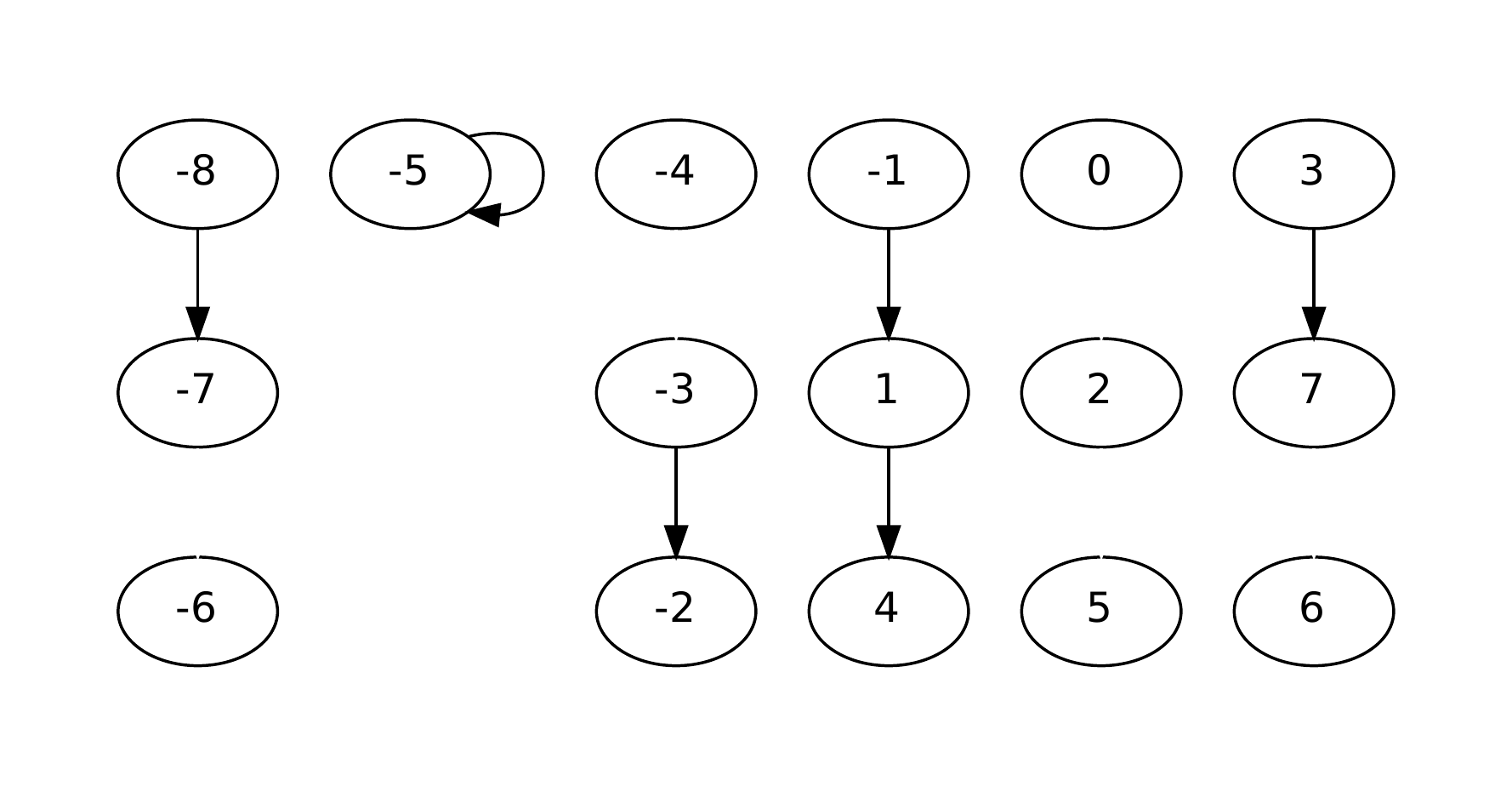}
\caption{Partition of [-8,7]}
\label{Partition}
\end{center}
\end{figure}
Each Connected Component(CC), $CC=p_{i}$, is a set of dependent iterations which should be executed sequentially. Different CCs can be executed parallely without synchronization.\\  
Partitioning the iteration space exposes all the parallelism naturally present in the loop. Here $|P| = 11$ is the maximum number of parallely executable CCs and the size of largest CC, max($|p_{i}|$) is 3. Among many possible semantically equivalent schedules a realizable 4-thread parallel schedule for the loop is:\\ 
\{\{-8,-7,-6,-5\},\{-3,-2,3,7\},\{-4,-1,1,4\},\{0,2,5,6\}\}.\\
This example though simple, brings out the possibility of Parallelizing a $VD^{3}$ loop.\\ 

\subsection{Background}
Our approach rests on using solutions of LDEs.\\
Without loss of generality, an LDE of the form $ax+by=c$, in this paper is assumed to be normalized where $a > 0$ and $a \le |b|$.\\
Given an Iteration space and a set of LDEs, we use the term \textbf{Precise dependence(PD)} to denote the collection of all solutions of all LDEs in the iteration space. A 2 variable LDE characterizes $VD^{3}$, except for the case $a=1$ and $b=-1$ where it characterizes $CD^{3}$.
Iterations which do not appear in PD are independent iterations.

\section{Precise Dependence for Single LDE}
A 2 variable LDE characterizes pair of 1 dimensional array access in a single non-nested loop. This section presents structure of P and efficient generation of PD for it. Consider the loop: 
\begin {verbatim}
for(i = LB ; i < UB ; i++)
{ S1 g[ai] = .........;
  S2 ......... = g[c - bi];}
\end{verbatim} 
Referring to figure \ref{Partition} and LDE $2x-3y = 5$, all solutions of the LDE can be derived using parametric solution \{-5 + 3t , -5 + 2t\}, where $-1 \le t \le 6$ is an integer\cite{exact}. If (x,y) is a solution of LDE and x precedes y in iteration order then x is called as source and y is called as sink. Using parametric solution it is possible to compute source of a given sink and sink of a given source.
Consider $source = 1$, $sink = -1$ is obtained for t = 2. Similarly using $sink = 1$, $source = 4$ is obtained. Hence (1,-1) and (4,1) are solutions of LDE. Since 1(source in one solution and sink in other) is an iteration common to both solutions \{-1,1,4\} is a CC. The as possible extensions \{8.5,4\} and \{-1,-2.33\} are not integer solutions so the size of $CC = 3$. Iteration 2 can be neither a sink nor a source, hence 2 is independent iteration.\\ \\
Observations based on aforementioned parametric solution of LDE $ax+by=c$ are:\\
\textbf{An Iteration which is a source as well as a sink extend a CC.}\\
\textbf{An Iteration which either source or sink but not both does not extend a CC.}\\
\textbf{An Iteration which neither source nor sink is and independent iteration forming singleton CC.}\\

\subsection{Formulation of CCs}
The objective of this section is to formulate the structure of a CC and relationship between its constituent iterations.
The solutions to LDE in parametric form $\{ (\alpha - mb ) , (\alpha + ma) \}$, where m is an integer and $\{ \alpha , \alpha\}$ is a particular solution form CCs of length 2.\\ \\
Longer CCs are formed if m is a multiple of a or b. As $(\alpha + mab)$ is a common iteration the CC is $\{ (\alpha - mbb ) , (\alpha + mab) , (\alpha -maa )\}$.\\ \\
The general structure of CCs for and LDE, where m is not a multiple of a or b is:\\
$\{ (\alpha \pm ma^{l}),(\alpha \mp ma^{l-1}b),(\alpha \pm ma^{l-2}b^{2}),.....,(\alpha \pm(-1)^{j}.ma^{i}b^{j}),.....,\\(\alpha \pm(-1)^{l-2}.ma^{2}b^{l-2}),(\alpha \pm(-1)^{l-1}.mab^{l-1}),(\alpha \pm(-1)^{l}. mb^{l})\}$\\ \\
\textbf{Seed is a representative iteration of a CC which can generate a CC.}
For a Single LDE single LOOP case all iterations are Seeds of corresponding CCs.If $\alpha$ is not an integer CC and Seeds are computed using particular solution $\{\beta , \gamma \}$ of LDE as: $\{ (\beta - mb ) , (\gamma + ma) \}$
\\ \\
\textbf{The Seeds for CCs can be computed using the aforementioned structure.}

\subsection{Bounds on $|P|$ and $|CC|$} 

This sub-section shows bounds on number $|P|$ and length $|CC|$ of CCs for LDE in a given range R = [L,U].\\ 

\begin{lem}
\label{arch_lemma_3}
Number of CCs of length $2$ is $R/|b|$.\\
\end{lem}
\begin{proof}
Every solution of LDE is a CC of length 2. The number of solutions in the range R is $R/|b|$. Hence proved.
\end{proof}

Number of Singleton CCs is $R- 2R/|b|$. 

\begin{lem}
Number of CCs of length $l$ (including sub-CCs of length $<l) = R/(|b|^{l-1})$.\\
\end{lem}

\begin{proof}
Proof by Induction: Number of CCs of length $2$ is $R/|b|$ using lemma \ref{arch_lemma_3}.\\
Assumption: Number of CCs of length $k$ is $R/(|b|^{k-1})$.\\
To Prove That: Number of CCs of length $k+1$ is $R/(|b|^{k})$.\\
An end element of a CC of length k has a form $(\alpha \pm ma^{k}) or (\alpha \pm mb^{k})$.\\ if $m= ka$ is a multiple of a the element $(\alpha \mp kb^{k+1})$ adds to the CC.\\
The maximum number of elements having pattern $(\alpha \mp kb^{k+1}) \in$ [LB,UB] is $R/(|b|^{k})$. Hence Proved.\\ 
\end{proof}

Let $p_{max}$ be a positive integer such that $R/(|b|^{P_{max+1}}) < 1 \le R/(|b|^{P_{max}})$ then the upper bound on length of CCs is $ L_{max} = P_{max} + 1$.\\
\\
The upper bound on Number of non singleton CCs is:
$N_{max} = R/|b| - R/(b^{2}) + R/(|b|^{3}) - R/(b^{4})...R/(|b|^{P_{max}})$\\
$= R\Big(|b|^{P_{max}}-(-1)^{P_{max}})/(|b|^{P_{max}}.(|b|+1)\Big)$\\
\\
This gives a closed-form expression for the bounds, a useful input for scheduling.

\section{Classification of loops for parallelization}
In multidimensional arrays, subscripts are separable if indices do not occur in other subscripts. A loop is separable 
if subscripts of all arrays in it (excluding arrays which are only read but not written) are separable.
We construct and prove precise dependence theory for simple loops first and then extend it to all separable loops.
Precise dependence information is obtained 
by solving LDE’s and obtaining the dependence graph for the given range of the loop.
A node from
each connected component in undirected dependence graph is selected and stored for
generation of the CC and scheduling purpose (here onwards referred as ’seed’).

Even though we present a general method for parallelizing and scheduling any type of separable loop, 
we demonstrate that some special cases have interesting properties which can be exploited for efficiency and effectiveness.


\subsection{Single Loop, Multiple LDEs}

Consider a loop $L_1$ having \emph{n} pairs of statements such that each pair has exactly one common array access \emph{(if a 
statement has more than one common array access, it can be broken into multiple statements having exactly one common 
array access)}, construct a LDE per pair
(construction is shown in previous section). Let the LDE's be LDE-1,LDE-2,LDE-3,
...LDE-n for $1^{st},2^{nd},3^{rd},...,n^{th}$ pair of statements respectively.

Let $S_1,S_2,S_3,...,S_n$ be the solutions of the LDE-1,LDE-2,LDE-3,...,LDE-n respectively.

\begin{equation}
S = S_1 \cup S_2 \cup S_3 \cup \ ...\  \cup S_n
\end{equation}

\begin{lem}
\label{lemma_1}
All the dependent iterations in $L_1$ are elements of S. (i.e if iterations $i_1$ and $i_2$ are dependent
then $(i_1,i_2) \in S$ or $(i_2,i_1) \in S$.
\end{lem}

\begin{proof}
Proof by contradiction.
Assume there is a pair of iteration $x_1$ and $x_2$ which are dependent but $(x_1,x_2) \notin S$ and
$(x_2,x_1) \notin S$.
$x_1$ and $x_2$ access same memory location of atleast one array (say array \emph{A}) and for some statements
$s_1$ and $s_2$. This implies $(x_1,x_2)$ or $(x_2,x_1)$ is a solution of LDE that represents $s_1$ and $s_2$.
Hence $(i_1,i_2) \in S$ or $(i_2,i_1) \in S$. 
\end{proof}

\begin{lem}
\label{lemma_2}
All elements in \emph{S} represent dependent iterations in $L_1$.
(i.e if $(i_1,i_2) \in S$ or $(i_2,i_1) \in S$ then iterations $i_1$ and $i_2$ are dependent.
\end{lem}

\begin{proof}
Proof by contradiction.
Assume there is an element $(x_1,x_2) \in S$ but $x_1$ and $x_2$ are not dependent.
If $(x_1,x_2) \in S$ then $(x_1,x_2) \in S_k$ for some $k^{th}$ pair of statements ($1 \le k \le n$)
and this implies $x_1$ and $x_2$ are dependent. 
\end{proof}

\begin{thm}
\label{PD_thm_single_loop}
 Precise dependency information for $L_1$ captures all parallelism present at iteration level granularity.
\end{thm}

\begin{proof}
 From lemma \ref{lemma_1} and \ref{lemma_2}. 
\end{proof}

When dependencies come from two or more LDE's, solutions combine to produce DAG's which represent the dependence graph of loop.
Seeds can be classified as, \textbf{common seeds:} seeds that occur in solutions of more than one LDE and 
\textbf{unique seeds:} seeds that occur in solution of exactly one LDE.

Dependence graph can be partitioned into connected components having atleast one node which is common seed for any of 
the $ \binom{n}{k}$  sets of \emph{k} LDE's each (where $1\le k\le n$).

	Common seeds of \emph{k} lde's can be obtained by intersection of common seeds of
	$\frac{k}{2}$ disjoint pair of lde's.

\subsection{Formulation of Common Seeds}
\label{formulation_of_common_seeds}
Common Seeds are iterations which are common to solutions of 2 or more LDEs. 
Consider solutions of 2 LDEs: $(source_{1},sink_{1})$ and $(source_{2},sink_{2})$. Common Seeds are formed when : $source_{1} = source_{2}$ or $source_{1} = sink_{2}$ or $sink_{1} = source_{2}$ or $sink_{1}= sink_{2}$\\ 
\\
Common Seeds are intersection of solutions of multiple LDEs.
\\
For n LDEs there are $^{n}C_{k}2^{k}$ ($1\le k \le n$) ways of forming common Seeds.

\subsection{Generating unique common seeds}
   A set of seeds is known as unique common seeds if all seeds in the set are common seeds
   and no two seeds are from the same connected component.

   Consider a loop in which dependent iterations are captured by $n$ LDE's, let CS be the set of common seeds 
   (obtained using method described 
   in \ref{formulation_of_common_seeds}), UCS be the set of unique common seeds, 
   $\forall i \in [n]$ $Seed_i$, $UL_i$ represent the seeds and unique seeds of $i^{th}$ LDE respectively.

	\begin{algorithm}
		  \label{unique_common_seeds_algorithm}
		  \caption{Algorithm for generating unique common seeds and unique seeds}
		  
	\begin{algorithmic}		  
		\State $ UCS = \O $   
		\For{i = 1 to n}
		    \State $UL_i = Seed_i$
		\EndFor
		  
		  \While{ $ CS \neq \phi $ }
		      \State pick an element \emph{e} from CS
		      \State $CS = CS - \{e\}$
		      \State $UCS = UCS \cup \{e\}$
		      \State $Part = schedule(e,\text{parametric $sol^n$ }$
		      \State  of LDE's) \footnote{discussed in \ref{schedule_algorithm}}
		      \State $CS = CS - Part$
		      
		      \For{i = 1 to n}
			\State $UL_i = UL_i - Part$
		      \EndFor
		  \EndWhile
		  \end{algorithmic}

		\end{algorithm}
   
\subsection{Nested loops with multiple lde's accessing only one common array}
We consider nested loops having no conditional statements and the body of loop is located inside the 
inner most loop. All the array access expressions are affine and are separable.

Consider a loop $L_2$ having \emph{l} number of nested loops, a common array $A$ of $n$ dimension, $m$ be the number of pair of 
array accesses for each array such that atleast one of the accesses in the pair is write.

Let $\forall a \in [l]$, $i_a$ be used in access expressions of $C_a$ number of subscripts and $D_a$ be a set 
which contains all the subscripts whose access expressions involve $i_a$.
$D_a=\{d_{a,1},d_{a,2},...,d_{a,C_a}\}$, where $\forall i \in [C_a], d_{a,i} \in [n]$. 

$[n]$ subscripts are partitioned such that.

\begin{eqnarray}
 \label{eqn_1}
 \forall i,j \in [l] \ \&\ i \neq j, D_i \cap D_j = \phi 
 \nonumber \\
 \bigcup_{\forall i \in [l]} D_i = [n]
\end{eqnarray}

Let $\forall u \in [m]$ and $\forall v \in [n]$, 
$S_{u,v}$ represent the solutions of the LDE formed from $u^{th}$ pair of array accesses on $v^{th}$ subscript.

Let
\begin{equation}
\label{eqn_2}
\forall u \in [m], \forall a \in [l], S_{u,a}^{1} =   
\begin{cases}
\bigcap_{\forall i \in D_a} S_{u,i} & \text{if }D_a \neq \phi
\\
R_a & \text{otherwise } 
\end{cases}
\end{equation}
where $R_a$ is the range of $a^{th}$ loop.

\begin{equation}
\label{eqn_3}
 Let\  \forall u \in [m], T_u = \prod_{\forall a \in [l]} S_{u,a}^{1}
\end{equation}

\begin{equation}
\label{eqn_4}
 Let\ T = \bigcup_{u=1}^{m} T_{u}
\end{equation}

\begin{lem}
 \label{lemma_10}
 T captures all the dependent pair of iterations in $L_2$. If iterations $(x_1,x_2,...,x_l)$ and $(y_1,y_2,...,y_l)$
 are dependent then
 \begin{eqnarray}
 \label{lemma_10_1}
  \prod_{\forall a \in [l]} \binom{(x_a,y_a) \ if\  D_i \notin \phi }{(e1,e2) | e1,e2 \in R_i  
  \ otherwise} \in T \nonumber \\
  OR \nonumber \\
  \prod_{\forall a \in [l]} \binom{(y_a,x_a) \ if\  D_i \notin \phi }{(e1,e2) | e1,e2 \in R_i  
  \ otherwise} \in T 
 \end{eqnarray}
 where $R_i$ is the range of $i^{th}$ loop
\end{lem}

\begin{proof}
 Proof by contradiction. Assume there are two iterations $(x_1,x_2,...,x_l)$ and $(y_1,y_2,...,y_l)$,
 which are dependent but don't satisfy equation \ref{lemma_10_1}. 
 
 Since the two iterations are dependent, atleast one pair of array accesses should access same memory location. 
 If the two iterations are accessing the same memory location, $\forall a \in [l], (x_a,y_a)$ has to
 simultaneously satisfy all the LDE's where $i_a$ is used in access expression.

 the two iterations access same memory location iff $(x_a,y_a)$ or $(y_a,x_a)$ satisfies all the lde's
 of subscripts in $D_a$ simultaneously.

\begin{eqnarray}
\label{lemma_10_2}
  \exists q \in [m], \forall a \in [l], (D_a \neq \phi) \implies 
  \nonumber \\
  \forall k \in D_a, (x_a,y_a) \in S_{q,k}
  \nonumber \\
  or
  \nonumber  \\
  \exists q \in [m], \forall a \in [l], (D_a \neq \phi) \implies 
  \nonumber \\
  \forall k \in D_a, (y_a,x_a) \in S_{q,k}
\end{eqnarray}

Equation \ref{lemma_10_2} and \ref{eqn_2} imply \ref{eqn_5}.

\begin{eqnarray}
  \label{eqn_5}
  \exists q \in [m], \forall a \in [l], (x_a,y_a) \in S_{q,a}^{1}
  \nonumber \\
\end{eqnarray}

Equation \ref{eqn_5}, \ref{eqn_3} and \ref{eqn_4} imply equation \ref{lemma_10_1} is true, but this
is a contradiction. Hence the proof.

\end{proof}

\begin{lem}
 \label{lemma_11}
 Every element of T represents atleast one pair of dependent iterations in $L_2$. If $ \prod_{i=1}^{l} \{(x_i,y_i)\} \in T $, 
 then iterations $(x_1,x_2,...,x_l)$ and $(y_1,y_2,...,y_l)$ are dependent.

\end{lem}

\begin{proof}
 Proof by contradiction. Assume there exists an element $E = \prod_{i=1}^{l} \{(x_i,y_i)\}$ and
 $E \in T$, but iterations $(x_1,x_2,...,x_l)$ and $(y_1,y_2,...,y_l)$ are not dependent.
 
 Equation \ref{eqn_4} implies \ref{eqn_6} and equation \ref{eqn_6} and \ref{eqn_3} imply \ref{eqn_7} .
 
 \begin{eqnarray}
  \label{eqn_6}
  \exists q \in [m], E \in T_q
  \\
  \label{eqn_7}
  (\exists q \in [m], \forall a \in [l], (x_a,y_a) \in S_{q,a}^{1} \nonumber
  \\
  or \nonumber
  \\
  \exists q \in [m], \forall a \in [r], (y_a,x_a) \in S_{q,a}^{1})
 \end{eqnarray}

 Equation \ref{eqn_6} and \ref{eqn_3} imply $\forall a \in [l]$, $(x_a,y_a)$ or $(y_a,x_a)$ satisfy all the 
 lde's present in $D_a$ simultaneously. If $D_a = \phi$ then $i_a$ does not appear in access expression of any subscript.
 Hence the two iterations are accessing atleast one common memory location and are dependent, but this is a 
 contradiction. Hence the proof.
\end{proof}

\begin{thm}
\label{theorem_9}
 Precise dependence information for $L_2$ captures all the parallelism 
 present at iteration level granularity.

\end{thm}

\begin{proof}
 From lemma \ref{lemma_10} and \ref{lemma_11}.
\end{proof}

\section{Nested loops and multiple multi-dimension arrays}

Consider a nested loop $L_3$ having \emph{l} number of loops and \emph{N} be the number of common arrays 
having $n_1,n_2,...,n_N$ dimensions
respectively and let $m_1,m_2,...,m_N$ be the number of pairs of statements accessing the common arrays respectively.
Let $P_1,P_2,...,P_N$ be the dependency information as computed by equation \ref{eqn_4} by considering one array at a time.

Let $P = \cup_{i=1}^{N} P_i $.

\begin{thm}
 Precise dependence \emph{P} for $L_3$ captures all the parallelism present.
\end{thm}

\begin{proof}
 $P_1,P_2,...,P_N$ captures all the iteration pairs which cause dependence in $1_{st},2^{nd},...,n^{th}$ array 
 respectively(from theorem \ref{theorem_9}). 
 Dependencies in loop arise because of dependencies in accessing atleast one common array. Hence \emph{P}
 captures all the dependencies present in the loop.
\end{proof}

\section{Generation of the connected component from seed during run-time}

A seed is a representative iteration of a CC. This section gives an algorithm for generating CC from its seed with its proof of correctness.

\subsection{Algorithm for generating iteration schedule}

Function: Computes a dependency preserving schedule of a connected component.

Input: A seed (belonging to the CC, whose schedule is being computed)
and solutions to the LDE's in parametric form.

Algorithm:

\begin{algorithm}
 \label{schedule_algorithm}
 \caption{Algorithm for generating iteration schedule}
\begin{algorithmic}

\State Add seed to min heap (referred to as heap in this algorithm).

\While{heap is not empty}
  \State remove minimum element \emph{e}
  \If{\emph{e} is not visited}
    \State Discover dependencies of \emph{e} and add 
	\State them to heap
    \If{there are no dependencies for \emph{e}}
      \State Add \emph{e} to schedule
    \Else
      \State Mark \emph{e} as visited 
      \State Add $e$ to heap
    \EndIf
  \Else
    \If { \emph{e} was not added to schedule}
      \State  Add \emph{e} to schedule
    \EndIf
  \EndIf  
\EndWhile

\end{algorithmic}
\end{algorithm}

\subsection{Correctness Proof}
Dependencies of a node are discovered by substituting the value of node in the solution of the LDE's. 
All nodes on which the given node is dependent on and all the nodes which depend on the given node are 
discovered (by theorem \ref{PD_thm_single_loop}).

Let $e$ be the minimum value node taken out of min heap, let $D_l$ be the set of nodes on which $e$
depends (i.e nodes in $D_l$ should be executed before $e$ and $\forall x \in D_l, x < e \ \&\ ((x,e) \in T \ or\  
(e,x) \in T )$), \emph{(T is constructed in equation \ref{eqn_4})} and $D_h$ be the set of nodes which
depend on $e$ (i.e nodes in $D_h$ should be executed after $e$ and $\forall x \in D_h, x > e \ \&\ 
((x,e) \in T \ or\  (e,x) \in T )$).

Since there are no cyclic dependencies, there is atleast one node for which $|D_l|=0$.

\begin{enumerate}
 \item If $|D_l| = 0 $ for $e$ then it can be executed and $D_h$ is added to heap.  
 
 \item If $|D_l| \neq 0 $ for $e$ and $e$ is not marked as visited then elements of $D_l$ and $D_h$ are added to heap
       along with $e$ marked as visited.Hence all the dependencies of $e$ have been added to heap.
 
 \item If $e$ is marked as visited then it can be executed. Since $e$ has been visited, nodes in $D_l$ of
       $e$ were added to heap previously. As $e$ is the minimum value node on heap now,
       this implies that all elements in $D_l$ of $e$ have been executed.
\end{enumerate}

Since the above mentioned process is repeated till heap is empty, all the nodes in the connected component
are scheduled/executed in dependency preserving order.

\subsection{Time and space complexity}
Since every node is examined atmost twice (i.e if $D_l \neq 0 $ for some node $e$ then $e$ is examined for
first time when $e$ is minimum value node on heap and $e$ is NOT marked as visited
and examined for second time to when $e$ is the minimum value
node on heap and it marked as visited). Maximum number of steps taken by the algorithm is atmost
$2n$, where $n$ is number of nodes in the connected component. Hence algorithm runs in $O(n)$ time.

As heap contains the nodes in the connected component, in the worst case, all but one of the nodes in the 
connected component will be on heap. Maximum space consumed by the algorithm is $n-1$ units.
Hence algorithm runs in $O(n)$ space.

\section{Iteration schedule for nested loops}
Consider a nested $n$ loop structure, let V contain the loops whose induction variables are used in array access expressions
and W contain the loops whose induction variables do not appear in any array access expression, 
$V \cup W = [n]$ and $V \cap W = \phi$.
Let loops be numbered as $l_1,l_2,...,l_n$ starting from outermost to innermost loop. 
Let $\forall i \in V, Seed_i$ be the set containing seeds to generate connected components of the partition arising out of all the LDEs involving $i^{th}$ loop induction variable.
We present Algorithm 3 for scheduling nested loops. It uses Algorithm 2 to schedule individual loops. Sequential nested loops will be executed by a group of synchronization free threads parallely. 

\begin{algorithm}
  \label{schedule_algorithm_nested_loop}
  \caption{Algorithm for generating iteration schedule for nested loops}
 \begin{algorithmic}
 \Procedure{schedule\_main}{}
  \For {i=1 to n} 
    \If { $l_i \in V$}
      \State Create $|Seed_i|$ number of threads.
      \State Do a one-to-one mapping 
	  \State of threads to seeds in $Seed_i$.
      \State Each of the created thread runs 
	  \State Algorithm 2 to generate a
	  \State schedule for running $i^{th}$ loop.
      \State Each thread calls 
	  \State SCHEDULE\_SUB(i+1).       
      \State exit
    \Else 
      \State Keep the $i^{th}$ loop sequential.
      \State Schedule consists of iterations in 
      \State range of $i^{th}$ loop.
    \EndIf
  \EndFor
  \State Execute body of loop. \footnote{If the loop is fully parallel and no loop induction variable's occur in array access expression}
  \EndProcedure

  \Procedure{schedule\_sub}{k}
  \For {i=k to n} 
    \If { $l_i \in V$}
      \State Create $|R_i|$ number of threads. \footnote{$R_i$ is the range of $i^{th}$ loop}
      \State Do a one-to-one mapping of 
      \State threads to values in $R_i$.                
    \Else 
      \State Keep the $i^{th}$ loop sequential.
      \State Schedule consists of iterations in 
      \State range of $i^{th}$ loop.
    \EndIf
  \EndFor
  \State Execute body of loop.
  \EndProcedure
  
 \end{algorithmic}

\end{algorithm}

\section{Experimental Results}
The analysis and algorithms developed and presented in this paper show how partitioning of iteration space can be done in case of $VD^3$ and multiple loops with multiple LDEs. We have considered 2 variable LDEs with separable loops. In the presence of multiple LDEs and the ensuing interactions between their individual partitions the generation of the overall partition of the system is difficult. While our formulation successfully captures the overall partition, the experimentation explores the extent of existing parallelism.
\\ \\
The experiments were parametrized by: i) Loop range ii) number of LDEs iii) Randomized coefficients for $VD_3$  and iv) number of experiments carried out. The loop range was varied from [ $-2$ to $2$,..., $-2^{15}$ to $2^{15}$], [1..90] LDEs and up to 100 experiments. For each experiment the partition was evaluated in terms of minimum, maximum and average number of CCs, which are measures of exploitable parallelism.
\begin{itemize}
\item Table \ref{table1lde} shows that the extent of parallelism on the average is atleast 95\% across various ranges.
Thus, $VD^{3}$ for single LDE has significant exploitable parallelism even for large loop ranges.
\begin{table}
\begin{center}
	\begin{tabular}{| p{0.9cm}| p{1.3cm}| p{1.1cm} |p{2.1cm}|}
	\hline
	Range & Parallels & CCs & Independents \\ \hline
	5 & 5 & 0 & 5 \\ \hline
	9 & 9 & 0 & 9 \\ \hline
	17 & 16.97 & 0.03 & 16.94 \\ \hline
	33 & 32.93 & 0.07 & 32.86 \\ \hline
	65 & 64.81 & 0.18 & 64.63 \\ \hline
	129 & 128.44 & 0.54 & 127.9 \\ \hline
	257 & 255.27 & 1.67 & 253.6 \\ \hline
	513 & 508.42 & 4.35 & 504.07 \\ \hline
	1025 & 1012.16 & 12.32 & 999.84 \\ \hline
	2049 & 2016.82 & 30.92 & 1985.9 \\ \hline
	4097 & 4021.72 & 72.44 & 3949.28 \\ \hline
	8193 & 8026.04 & 161.06 & 7864.98 \\ \hline
	16385 & 16036.87 & 335.95 & 15700.92 \\ \hline
	32769 & 32058.47 & 685.59 & 31372.88 \\ \hline
	65537 & 64101.76 & 1384.98 & 62716.78 \\ \hline
 	\end{tabular}
	\caption{Average exploitable parallelism for single LDE}
	\label {table1lde}
\end{center}
\end{table}
\item Table \ref{tablemanylde} shows average parallelism when the number of LDEs are increased up to 90. The number of parallel components reduce considerably when the number of LDEs increase. This is a consequence of ineraction between the partitions of the individual LDEs. Figure \ref{paragraph} shows the range of overall exploitable parallelism for 1,30 and 90 LDEs. 
The table  \ref{tablemanylde} shows average behaviour and graph shows maximum parallelism. It can be seen that as the as the number of LDEs increse the maximun paralleism decreses. There is not enough differnece in the minimum parallelism is 
\begin{table}
\begin{center}
	\begin{tabular}{| p{0.9cm}| p{1.3cm}| p{1.3cm} |p{1.3cm}|p{1.1cm}|}\hline
	Range & 2LDE & 5LDEs & 30LDEs & 90LDEs \\ \hline
	5 & 4.97 & 4.93 & 4.61 & 3.99 \\ \hline
	9 & 8.92 & 8.89 & 8.16 & 6.78 \\ \hline
	17 & 16.82 & 16.69 & 15.11 & 11.74 \\ \hline
	33 & 32.59 & 32.39 & 28.57 & 20.66 \\ \hline
	65 & 64.15 & 63.28 & 54.06 & 34.68 \\ \hline
	129 & 127.11 & 124.20 & 99.11 & 51.44 \\ \hline
	257 & 252.79 & 244.76 & 176.12 & 68.40 \\ \hline
	513 & 501.27 & 481.27 & 308.30 & 87.65 \\ \hline
	1025 & 993.81 & 945.15 & 525.85 & 114.74 \\ \hline
	2049 & 1972.04 & 1852.66 & 897.32 & 150.58 \\ \hline
	4097 & 3920.22 & 3646.23 & 1557.49 & 211.95 \\ \hline
	8193 & 7810.89 & 7224.94 & 2853.74 & 346.77 \\ \hline
	16385 & 15599.41 & 14388.60 & 5488.09 & 628.31 \\ \hline
	32769 & 31176.54 & 28714.82 & 10722.84 & 1159.05 \\ \hline
	65537 & 62330.68 & 51724.69 & 21179.73 & 2199.60 \\ \hline
\end{tabular}
	\caption{Average exploitable parallelism for varying No. of LDEs}
	\label {tablemanylde}
\end{center}
\end{table}

\begin{figure}[!htb]
\begin{center}
\includegraphics[scale = 0.6]{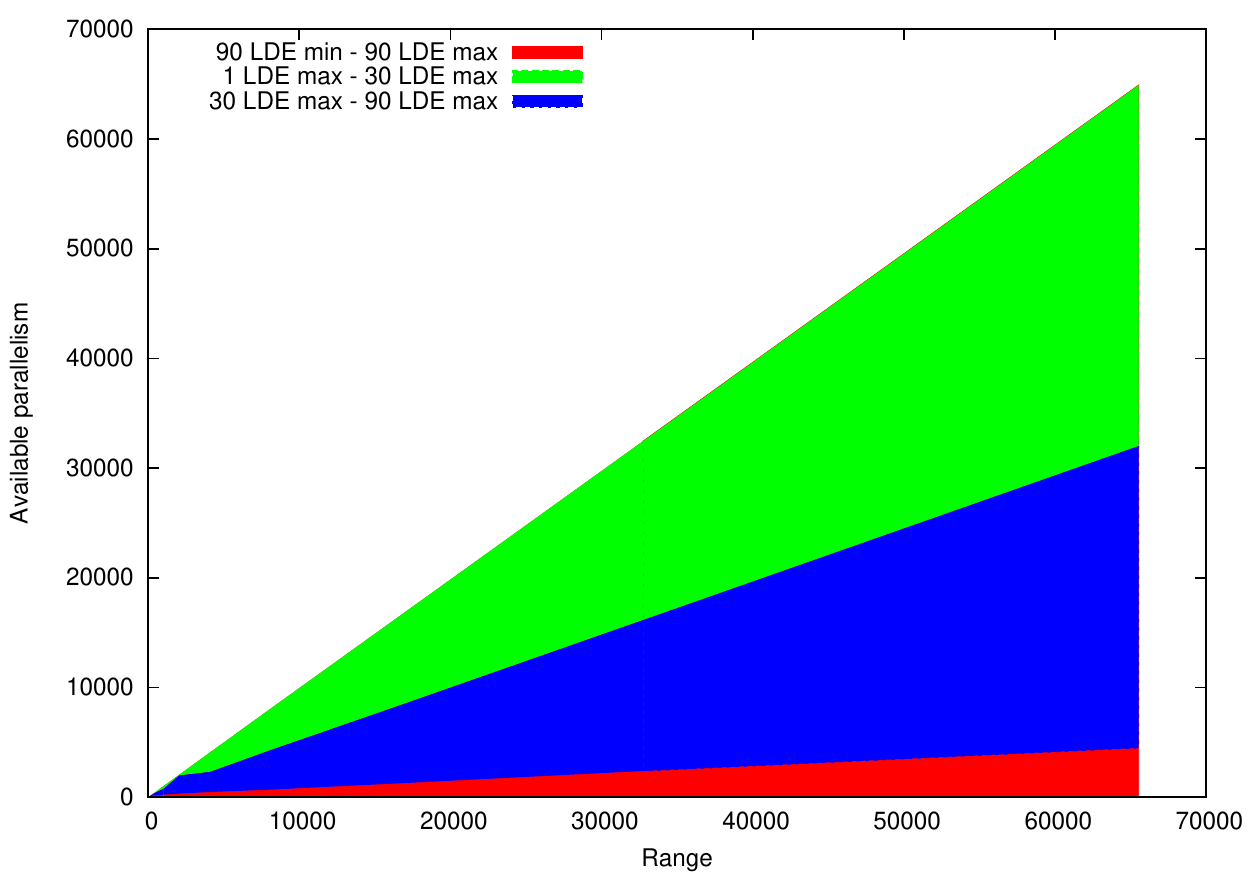}
\caption{Overall Exploitable parallelism}
\label{paragraph}
\end{center}
\end{figure}

\item Results show that exploitable parallelism decreases for increasing number of loops. It was observed that a large number of dependences cross loop independent iterations ($\alpha$s) in solutions of LDEs. We computed a representative iteration $\alpha$r using heuristic based on set of loop independent iterations ($\alpha$s ). The range was split into [LB,$\alpha$r] and [$\alpha$r+1,UB]. The partition 
for all the LDEs reported in table \ref{tablemanylde} have been recreated using the heuristic. This is reported in table \ref{tablealphacut} it shows a significant improvement in exploitable parallelism if the range is split into [LB , $\alpha_{h}$] and [$\alpha_{h} + 1$ , UB]. The improvement ratio increases to approximately 20:1 as the number of LDEs and range increase. Figure \ref{alphagraph} shows significant improvement in maximum parallelism for 1,30 and 90 LDEs after splitting the range into 2 ranges. Comparing Figure \ref{alphagraph} and Figure \ref{paragraph} it can be observed that the effective maximum parallelism shows considerable improvement. To summarize the heuristic used by us to split the iteration ranges seems to profitable.
\end{itemize}
\begin{table}
\begin{center}
	\begin{tabular}{| p{0.9cm}| p{1.3cm}| p{1.3cm}| p{1.1cm}| p{1.3cm}|} \hline
     & \multicolumn{2}{c|}{ 30LDEs } & \multicolumn{2}{c|}{ 90LDEs} \\ \hline
	Range & Avg & split-Avg & Avg & split-Avg \\ \hline
	5 & 4.61 & 4.68 & 3.99 & 4.15 \\ \hline
	9 & 8.16 & 8.49 & 6.78 & 6.93 \\ \hline
	17 & 15.11 & 15.57 & 11.74 & 12.23 \\ \hline
	33 & 28.57 & 29.29 & 20.66 & 21.44 \\ \hline
	65 & 54.06 & 56.10 & 34.68 & 37.73 \\ \hline
	129 & 99.11 & 108.25 & 51.44 & 70.25 \\ \hline
	257 & 176.12 & 217.2 & 68.40 & 144.19 \\ \hline
	513 & 308.30 & 439.07 & 87.65 & 318.07 \\ \hline
	1025 & 525.85 & 880.85 & 114.74 & 660.92 \\ \hline
	2049 & 897.32 & 1754.99 & 150.58 & 1330.00 \\ \hline
	4097 & 1557.49 & 3514.13 & 211.95 & 2651.60 \\ \hline
	8193 & 2853.74 & 7065.49 & 346.77 & 5410.43 \\ \hline
	16385 & 5488.09 & 14264.31 & 628.31 & 11192.73 \\ \hline
	32769 & 10722.84 & 28702.30 & 1159.05 & 22849.37 \\ \hline
	65537 & 21179.73 & 57578.68 & 2199.60 & 46151.75 \\ \hline
\end{tabular}
	\caption{Improved Average exploitable parallelism after range splitting}
	\label {tablealphacut}
\end{center}
\end{table}

\begin{figure}[!htb]
\begin{center}
\includegraphics[scale = 0.6]{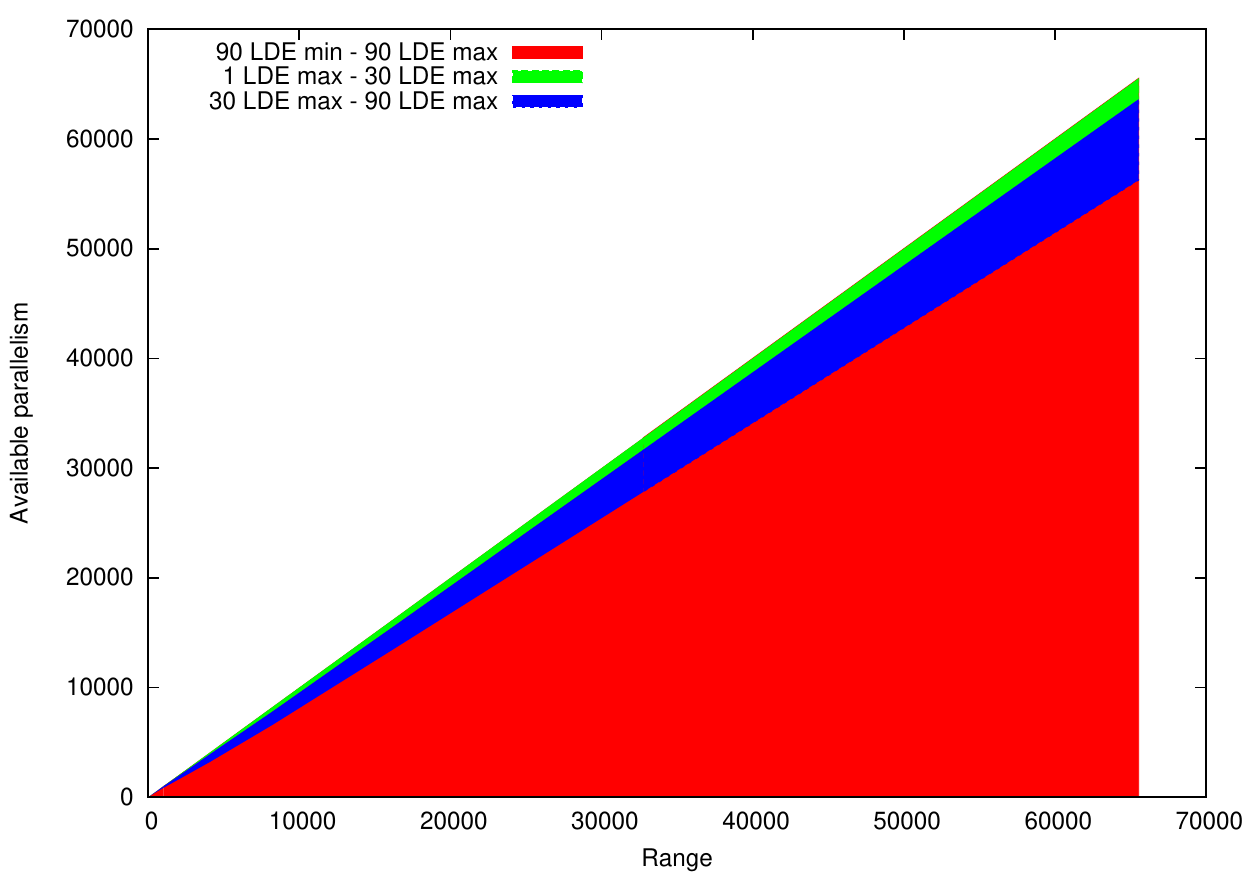}
\caption{Exploitable parallelism after cut}
\label{alphagraph}
\end{center}
\end{figure}

\section {Conclusion}

Capturing dependence in case of loops with $VD^{3}$ is presented in this paper. Our approach uses a two variable LDE and its parametric solutions to analyze and mathematically formulate Precise Dependence, the inherent dependence in the loops. Theoretical methods to compute precise dependence are presented. The result is a partition of the iteration space, a set of parallel CCs. We computed bounds on the number and size of CCs. The partition is represented in a reduced form by a set of seeds. Algorithms to use this set for generation of the components dynamically in a demand driven manner which can lead to schedules are presented.
As we could not find loops with large number of LDEs in practice randomly created multiple LDEs were used for our experimentation. Results show existence of reasonable parallelism which reduces as the number of LDEs increase. To improve results we have formed an alternate partition using a heuristic. This heuristic shows non trivial improvement in parallelism. Further method to obtain synchronization free parallelism is presented at iteration level granularity using efficient algorithms.  
\\ \\
While computing precise dependence in the most general case is NP-complete, most of the array access patterns found in practice are amenable to our approach. Our handling of exploitable parallelism in $VD_3$ and multiple loops expands the domain of applications that are parallelizable today. Compile time analysis for precise dependence opens opportunity of automatic generation of parallel threads suited for multi-core configuration. Even heavily connected dependence graphs show considerable parallelism which can be exploited by a well designed scheduling algorithm.
\\ \\
Some of the possible future directions are: characterizing class of loops for which precise dependence can be computed. Integrating precise dependence analysis, seed generation and scheduling for automatic generation of parallel code.


\end{document}